\documentclass[11pt]{article}
\usepackage{fullpage}
\usepackage{amsmath,amsthm, amssymb}
\usepackage{algpseudocode,algorithm}
\newtheorem{thm}{Theorem}

\newtheorem{lem}[thm]{Lemma}

\newtheorem{prop}[thm]{Proposition}
\theoremstyle{definition}
\newtheorem{defn}[thm]{Definition}
\theoremstyle{remark}

\newcommand{\R}{\mathbb R}

\newcommand{\I}{\mathcal{I}}

\newcommand{\M}{\mathcal M}
\newcommand{\C}{\mathcal C}
\newcommand{\F}{\mathcal F}

\begin{document}
\title{Symmetric Submodular Function Minimization \\ Under Hereditary
  Family Constraints\footnote{This work was partially supported by NSF
  contract CCF-0829878 and by ONR grant N00014-05-1-0148.}}

\author{Michel X. Goemans\thanks{MIT, Dept.~of Math., Cambridge, MA
02139. \texttt{goemans@math.mit.edu}.} \and Jos\'e A. Soto\thanks{MIT,
Dept.~of Math., Cambridge, MA 02139. \texttt{jsoto@math.mit.edu}.}}
\date{~} \maketitle

\thispagestyle{empty}

\begin{abstract} We present an efficient algorithm to find non-empty
minimizers of a symmetric submodular function over any family of sets
closed under inclusion. This for example includes families defined by
a cardinality constraint, a knapsack constraint, a matroid
independence constraint, or any combination of such constraints. Our
algorithm make $O(n^3)$ oracle calls to the submodular function where
$n$ is the cardinality of the ground set. In contrast, the problem of
minimizing a general submodular function under a cardinality
constraint is known to be inapproximable within $o(\sqrt{n/\log n})$
(Svitkina and Fleischer [2008]).

The algorithm is similar to an algorithm of Nagamochi and
Ibaraki [1998] to find all nontrivial inclusionwise minimal minimizers
of a symmetric submodular function over a set of cardinality $n$ using
$O(n^3)$ oracle calls. Their procedure in turn is based on
Queyranne's algorithm [1998] to minimize a symmetric submodular
function.
\end{abstract}

\newpage
\setcounter{page}{1}
\section{Introduction}

Consider a finite set $V$, and a real set function $f:2^V \to \R$ on
$V$.  Given two different sets
$A,B\subseteq V$, we say that $A$ and $B$ are \emph{crossing} if
$A\setminus B$, $B\setminus A$, $A\cap B$ and $V\setminus (A\cup B)$
are all non-empty.
The function $f$ is \emph{submodular} (resp.~\emph{crossing
submodular}) over $V$ if
\begin{equation}\label{eqn-submodularity}
f(A \cup B) + f(A \cap B) \leq f(A) + f(B),
\end{equation}
for every pair of subsets (resp.~crossing subsets) $A$ and $B$ of
$V$. Observe that any submodular function is also crossing submodular,
by definition. 
A pair
$(V,f)$ where $f$ is (crossing) submodular  is called a
\emph{submodular system}. The function $f$ is further called
\emph{symmetric} if
\begin{equation}\label{eqn-symmetry}
f(A) = f(V\setminus A), \text{ for all $A \subseteq V$}.
\end{equation}

Submodularity is observed in a wide family of problems. The rank
function of a matroid, the cut function of a (weighted, directed or
undirected) graph, the entropy of a set of random variables, or the
logarithm of the volume of the parallelipiped formed by a set of
vectors are all examples of submodular functions. Furthermore, many
combinatorial optimization problems can be formulated as minimizing a
submodular function; this is for example the case for the problem of
finding the smallest number of edges to add to make a graph
$k$-edge-connected. Therefore, the following problem is considered a
fundamental problem in combinatorial optimization.  

\paragraph{Unconstrained Minimization Problem:} Given a submodular
system $(V,f)$, find a subset $A^* \subseteq V$ that minimizes
$f(A^*)$.\\

A submodular function $f$ is usually given by an oracle which, given a set
$S$, returns $f(S)$. Gr\"otschel, Lov\'asz and
Schrijver~\cite{grtschel_ellipsoid_1981,grtschel_geometric_1993} show that
this problem can be solved in strongly polynomial time using the
ellipsoid method; by running time we mean both the computation time
and the number of oracle calls. Later, a collection of combinatorial strongly
polynomial algorithm have been developed by several
authors~\cite{fleischer_push-relabel_2003,iwata_combinatorial_2001,iwata_fully_2002,orlin_faster_2007,schrijver_combinatorial_2000,iwata_simple_2009}. The
fastest purely combinatorial algorithms known so far, due to Iwata and
Orlin~\cite{iwata_simple_2009} and Orlin~\cite{orlin_faster_2007} make
$O(|V|^5\log(|V|))$ and $O(|V|^5)$ function oracle calls respectively.

When $f$ has more structure, faster algorithms are known. The case where $f$ is symmetric is of special interest. In
this case, we also require the minimizer $A^*$ of $f$ to be a
\emph{nontrivial} subset of $V$, that is $\emptyset \subset A^*
\subset V$, otherwise the problem becomes trivial since, by symmetry
and submodularity, $f(\emptyset)= \frac12(f(\emptyset) + f(V)) \leq
\frac12 (f(A) + f(V \setminus A)) = f(A)$, for all $A\subseteq V$.

The canonical example of a symmetric submodular function is the cut
capacity function of a nonnegatively weighted undirected
graph. Minimizing such a function corresponds to the minimum cut
problem. Nagamochi and Ibaraki~\cite{nagamochi_computing_1992,
  nagamochi_linear-time_1992} give a combinatorial algorithm to
solve this problem without relying on network flows. This algorithm
has been improved and simplified independently by Stoer and
Wagner~\cite{stoer_simple_1997} and
Frank~\cite{frank1994edge}. Queyranne~\cite{queyranne_minimizing_1998}
generalizes this and obtains a purely combinatorial algorithm
that minimizes a symmetric submodular function using only $O(|V|^3)$
function oracle calls.

In this paper, we are interested in the problem of minimizing
symmetric submodular functions over subfamilies of $2^V$ that are
closed under inclusion. More precisely, an \emph{hereditary family}
$\I$ (also called a lower ideal, or a down-monotone family) over $V$
is defined as a collection of subsets of $V$ such that if a set is in
the family, so are all its subsets. A triple $(V,f,\I)$ where $f$ is
symmetric and submodular on $V$, and $\I$ is an hereditary family on
$V$ is called an \emph{hereditary submodular system}. A natural
version of the minimization problem associated to hereditary
submodular systems is the following.

\paragraph{Hereditary Minimization Problem:} Given an hereditary
submodular system $(V,f,\I)$, find a subset $\emptyset \neq A^* \in
\I$ that minimizes $f(A)$ over all the sets $A \in \I$.\\

Common examples of hereditary families include
\begin{itemize}
  \item \textbf{Cardinality families:} For $k\geq 0$, the family of all subsets with at most $k$ elements: $\I=\{A \subseteq V: |A| \leq k\}$.
  \item \textbf{Knapsack families:} Given a weight function $w: V \to \R_+$, consider the family of all subsets of weight at most one unit: $\I=\{A \subseteq V: \sum_{v \in A} w(v) \leq 1\}$.
  \item \textbf{Matroid families:} Given a matroid $\M$ over $V$,
consider the family of independent sets of $\M$.
  \item \textbf{Hereditary graph families:} Given a graph $G=(V,E)$,
consider the family of sets $S$ of vertices such that the induced
subgraph $G[S]$ satisfies some hereditary property such as being a
clique, being triangle-free, being planar (or exclude certain minors).    
  \item \textbf{Matching families:} Given a hypergraph $H=(V,E)$, consider the family of matchings of $H$, that is sets of edges that are pairwise disjoint.
\end{itemize}
The hereditary minimization problem includes, for example, the problem
of finding a planar induced subgraph in an undirected graph minimizing
the number (or the weight) of edges in its coboundary (i.e.~with
precisely one endpoint in the set). 

Noting that the intersection of hereditary families is also hereditary
we can see that the previous minimization problem is very general. In
fact, for general submodular functions, this problem cannot be
approximated within $o(\sqrt{|V|/\log |V|})$ using a polynomial number
of queries even for the simpler case of cardinality families (see
\cite{svitkina_submodular_2008}). In this paper we focus on the
symmetric case, extending Queyranne's algorithm as follows.

\begin{thm}\label{thm-sym} Given a symmetric and crossing submodular
function $f$ on $V$, and an hereditary family $\I$ of subsets of $V$,
an optimal solution for the associated hereditary minimization problem
can be found using $O(|V|^3)$ function value oracle calls.  \end{thm}

In this statement, an \emph{optimal solution} refers to a nonempty set
$A^* \in \I$ that attains the minimum in the hereditary minimization problem. Our algorithm in fact returns a
\emph{minimal} solution among all optimal solutions, that is one such
that no proper subset of it is also optimal. 

For the unrestricted problem, Nagamochi and
Ibaraki~\cite{nagamochi_noteminimizing_1998}
present a modification of
Queyranne's algorithm that finds all inclusionwise minimal minimizers
of a symmetric submodular function still using a cubic number of
oracle calls. Using similar ideas, we can also list all minimal
solutions of an hereditary minimization problem using only $O(|V|^3)$
oracle calls. As these minimal solutions can be shown to be disjoint,
there are at most $|V|$ of them.

\begin{thm}\label{thm-minimal} 
Given a symmetric and crossing submodular function $f$ on $V$, and an
hereditary family $\I$ of subsets of $V$, the collection of all
minimal optimal solutions for the associated hereditary minimization
problem can be found using $O(|V|^3)$ function value oracle calls.
\end{thm}

Finally, we also give some general conditions for other classes of
functions for which our methods can still be applied, see Section
\ref{sec:ext}. For instance, we can find all the minimal minimizers of
a function $f$ under hereditary constraints when $f$ is a
\emph{restriction} of a symmetric submodular function (also known as
\emph{submodular-posimodular functions}) or when $f(S)$ is defined as
$d(S,V\setminus S)$ for a monotone and consistent symmetric set map
$d$ in the sense of Rizzi~\cite{rizzi_note_2000}. See section
\ref{sec:ext} for definitions and precise statements. An example of
the latter setting is to find an induced subgraph $G[S]$ satisfying
certain hereditary property (e.g.~being planar or bipartite) and
minimizing the maximum (weighted) distance between any vertex in $S$
and any vertex in $V\setminus S$ (and this does not define a submodular
function).   

\paragraph{Other related work.}

Constrained submodular function minimization problems, i.e.~the
minimization of a submodular function over \emph{subfamilies} of
$2^V$, have also been studied in different contexts. Padberg and Rao
\cite{PadbergR82} show that the minimum odd cut problem obtained by
restricting the minimization over all odd sets can be solved in
polynomial time. This was generalized to submodular functions over
larger families of sets (satisfying certain axioms) by Gr\"otschel,
Lov\'asz and Schrijver \cite{grtschel_geometric_1993} and by Goemans and
Ramakrishnan~\cite{goemans_minimizing_1995}. This covers for example
the minimization over all even sets, or all sets not belonging to a
given antichain, or all sets excluding all minimizers (i.e.~to find
the second minimum).   For the particular case of
minimizing a \emph{symmetric} submodular function under cardinality
constraints the best previous result is a 2-approximation algorithm by
Shaddin Dughmi \cite{Dughmi2009}. Recently, Goel et
al~\cite{goel_submodular2009} have studied the minimization of
\emph{monotone} submodular functions constrained to sets satisfying
combinatorial structures on graphs, such as vertex covers, shortest
paths, perfect matchings and spanning trees, giving inapproximability
results and almost matching approximation algorithms for
them. Independently, Iwata and Nagano~\cite{Iwata_covering2009} study
both the vertex and the edge covering version of this problem.

The algorithm of Nagamochi and
Ibaraki~\cite{nagamochi_noteminimizing_1998} also works
with functions satisfying a less restrictive symmetry
condition. Narayanan~\cite{narayanan_noteminimization_2003}
shows that Queyranne's algorithm can be used to minimize a wider
class of submodular functions, namely functions that are contractions
or restrictions of symmetric submodular
functions. Rizzi~\cite{rizzi_note_2000} has given  further extension of this
algorithm for a different class of functions.

\section{Pendant pairs and Queyranne's algorithm}

In this section we review Queyranne's algorithm for the unconstrained
minimization problem of a system $(V,f)$, where $f$ is symmetric and
crossing submodular. An ordered pair $(t,u)$ of elements of $V$ is
called a \emph{pendant pair} for $(V,f)$ if $\{u\}$ has the minimum
$f$-value among all the subsets of $V$ containing $u$ but not $t$,
that is:
\begin{equation}
f(\{u\}) = \min\{f(U): U\subset V, t\not\in U \text{ and } u \in U\}.
\end{equation}

We can find a pendant pair by constructing an ordering $v_1,\ldots,v_n$ of the elements of $V$, with $|V|=n$, such that 
\begin{equation}
  f(W_{i-1} + v_i) - f(v_i) \leq f(W_{i-1} + v_j) - f(v_j), \text{ for all } 2\leq i \leq j \leq n, \label{minlex}
\end{equation}
where $v_1$ can be chosen arbitrarily, and $W_i$ denotes the set
$\{v_1,\ldots,v_i\}$. In the above inequality, we have used the
notation $W+v$ for $W\cup \{v\}$. An order successively satisfying
(\ref{minlex}) is called a \emph{legal
order}. Queyranne~\cite{queyranne_minimizing_1998} shows the
following result:
\begin{lem} \label{lem-pendantpair} 
For a symmetric and crossing submodular function $f$ on $V$, and an
arbitrarily chosen element $v_1 \in V$, the last two elements
$(v_{n-1},v_n)$ of a legal order of $V$ starting from $v_1$ constitute
a pendant pair. Furthermore, this legal ordering can be found by using
$O(|V|^2)$ function value oracle calls.
\end{lem}

Queyranne proves this for a symmetric and submodular function
$f$. However, as observed by Nagamochi and
Ibaraki~\cite{nagamochi_noteminimizing_1998}, this proof only requires
symmetry and crossing submodularity, so Lemma~\ref{lem-pendantpair}
holds.

Observe that if $(t,u)$ is a pendant pair for a symmetric and crossing
submodular function $f$, and $X^*$ is an optimal set, then either
$X^*$ separates $t$ from $u$, in which case $\{u\}$ must also be an
optimal solution, or $X^*$ does not separate $t$ from $u$. In the
latter case we can contract the pair $t$ and $u$ into $t$ (for
simplicity, we reuse $t$), this is, consider the symmetric and
crossing submodular function $f'$ on $V' = V\setminus \{u\}$ defined
by:
\begin{align} \label{eqn-contraction-f}
  f'(X) &= \begin{cases}
    f(X), &\text{if } t \not\in X\subseteq V'\\
    f(X + u), &\text{if } t \in X \subseteq V'.
  \end{cases}
\end{align}
Now (still assuming that not all minimums of $f$ separate $t$ and
$u$), we can obtain an optimal solution $X^*$ for $f$ from an optimal
solution $\hat{X}$  for $f'$.  If $t\notin \hat{X}$, we set $X^*=\hat{X}$, while
if $t\in \hat{X}$, we set $X^*=\hat{X}+u$. Applying this argument $n-1$ times and Lemma \ref{lem-pendantpair}, one can find an optimal solution for the original function by using $O(|V|^3)$ function value oracle calls.

By exploiting the fact that the first element in a legal order can be
chosen arbitrarily, we modify the above argument to also work in the
hereditary version.  In order to do this, it is useful to extend the
notion of contraction as follows.

\begin{defn}Given an hereditary submodular system $(V,f,\I)$, an element $t\in V$ and a set of elements $L \subseteq V$ containing $t$, the system $(V',f',\I')$  obtained by \emph{contracting $L$ into $t$} is defined as follows
\begin{align}\label{eqn-contraction}
V' &= V\setminus L + t;\\
f'(X) &= \begin{cases}
    f(X), &\text{if } t \not\in X\subseteq V'\\
    f(X \cup L), &\text{if } t \in X \subseteq V';
  \end{cases}\\
\I' &= \{X \subseteq V':\ t \in X, X \cup L \in \I\} \cup \{X \subseteq V':\ t \not\in X, X \in \I\}.
\end{align}
\end{defn}

It is easy to check that this construction preserves submodularity
(even crossing submodularity) and symmetry, and that the new family
$\I'$ is also hereditary.  In the next section we use this notion of
contraction iteratively, so it is useful to explore some of its
properties. We associate to each element $w$ of the ground set of a
particular iteration the subset $X_w$ of elements in the original
ground set that have been contracted to it so far. It is easy to check
that a set $A$ of elements in the current ground set belongs to the
contracted hereditary family if and only if the set $X_A = \bigcup_{w
\in A} X_w$ is a member of the original hereditary family, and in
fact, for every set $A$ in the contracted family, $f'(A)=f(X_A)$.  We
also need some extra notions. An element $v \in V$ such that $\{v\}
\not\in \I$ is called a \emph{loop} of $\I$. In particular, if $s$ is
a loop in the original family and we contract some elements into $s$
then, in the resulting contracted family, $s$ is still a loop (by the
hereditary property). Also, for any two (possibly contracted) elements
$t$ and $u$ of $V'$, we say that a set $X \subseteq V$ \emph{separates} $t$
and $u$ if $X_t \subseteq X$ and $X_u \subseteq V\setminus X$ or
vice versa.

\section{Hereditary minimization problem}\label{sec-hereditary}

In what follows, assume that $f$ is a symmetric and crossing
submodular function on $V$, and $\I$ is a non-trivial hereditary
family (i.e. $V \not\in \I$). We show how to compute all minimal
optimal solutions of the hereditary minimization problem given by
$(V,f,\I)$.

Note that if $X$ and $Y$ are two minimal solutions in $\I$ that cross
then $X\setminus Y$ and $Y\setminus X$ are also in $\I$ and, by
minimality, we have $f(X\setminus Y)> f(X)$ and $f(Y\setminus X)
>f(Y)$, implying that $f(X\setminus Y) + f(Y\setminus X) > f(X) +
f(Y)$. Since $X$ and $Y$ cross, the sets $X$ and $V\setminus Y$ are
also crossing and so, using the symmetry and crossing submodularity of
$f$ we also get 
$$f(X\setminus Y) + f(Y\setminus X) \leq  f(X) + f(Y), $$
contradicting our assumptions.  Now suppose that $X$ and
$Y$ are minimal solutions that intersect but not cross (i.e. $X \cup
Y=V$), then by symmetry $V\setminus X \subset Y$ and $V\setminus Y
\subset X$ are also optimal solutions, contradicting the minimality of
the original sets. The previous discussion implies that all minimal
solutions are pairwise disjoint and, in particular, there are only a
linear number of them.

\begin{prop}
Given a symmetric and crossing submodular function $f$ on $V$, and an
hereditary family $\I$ of subsets of $V$, the collection of all
minimal optimal solutions for the associated hereditary minimization
problem are disjoint. 
\end{prop}

In what follows we present two algorithms, one to find a particular
minimal optimal solution of the system and another to find all of
them. We remark here that both algorithms are direct extensions of the
algorithms presented by Nagamochi and
Ibaraki~\cite{nagamochi_noteminimizing_1998}, and in fact if we set
$\I$ to be the hereditary family of sets not containing a particular
element $s$, we recover their algorithms.

Since $f$ is crossing submodular we can use Queyranne's lemma to find
a pendant pair $(t,u)$, keep the set associated to $u$ as a candidate
for the optimal solution, contract the pendant pair and
continue. However, this might introduce candidates that are not in the
original hereditary family.  In order to avoid that, we first contract
all the loops, if any, of $\I$ into a single loop $s$, and proceed to
find a pendant pair not containing it, by using $s$ as the \emph{first
element} of the legal order. In this way, we can ensure that every
candidate for optimal solution belongs to $\I$. If the hereditary
family has no loops, then we simply use Queyranne's procedure until a
loop $s$ is created. From that point on we continue as before. The
complete procedure is depicted in Algorithm \ref{algo}.

\begin{algorithm}[h!!!]
\caption{FindOptimal $(V,f,\I)$}
\label{algo}
\begin{algorithmic}[1]
\Require{A submodular system $(V,f,\I)$ where $f$ is symmetric and crossing submodular, and $\I$ is not trivial.}
\Ensure{An optimal set $X^*$ for the hereditary minimization problem.}
\State{Let $(V',f',\I') = (V,f,\I)$, and $\C=\emptyset$. \Comment{$\C$ is the set of candidates.}}
\While{$\I$ has no loops}
        \State{Find any pendant pair $(t,u)$ of $f'$.}
        \State{Add $X_u$ to $\C$. \Comment{$X_u$ is the set of elements of $V$ that have been contracted to $u$. }}
        \State{Update $(V',f',\I')$ by contracting $\{t,u\}$ into $t$.}
\EndWhile \Comment{$\I'$ has at least one loop.}
\State {Let $(V'+s,f',\I')$ be the system obtained by contracting all the loops of $\I$ into $s$ (during the rest of the algorithm, we keep $s$ as an element \emph{outside} $V'$)}
\While{$|V'|\geq 2$}
    \State{Find a pendant pair $(t,u)$ of $f'$ not containing $s$.}
    \State{Add $X_u$ to $\C$.}
    \If{$\{t,u\} \in \I'$}
        \State{Update $(V'+s,f',\I')$ by contracting $\{t,u\}$ into $t$.}
    \Else
        \State{Update $(V'+s,f',\I')$ by contracting $\{s,t,u\}$ into $s$.}
    \EndIf
\EndWhile
\If{$|V'|=1$ (say $V'=\{t\}$)}
    \State{Add $X_t$ to $\C$.} \label{line-last}
\EndIf
\State{Return the set $X^*$ in $\C$ with minimum $f$-value that was added first to $\C$.}
\end{algorithmic}
\end{algorithm}

\begin{thm}\label{thm-sym-alg} Given a symmetric and crossing
submodular function $f$ on $V$, and an hereditary family $\I$ of
subsets of $V$, Algorithm \ref{algo} outputs a minimal optimal
solution for the associated hereditary minimization problem in
$O(|V|^3)$ function value oracle calls.  \end{thm}

Let us check the correctness of the algorithm. By induction we can check that at the beginning of each iteration, either $\I'$ is loopless or $s$ is its only loop. From here, we get that the element $u$ of the pendant pair $(t,u)$ found by the algorithm is not a loop of $\I'$, and thus, every candidate set $X_u$ is an element of the original hereditary family $\I$.

To check optimality of $X^*$, we claim that if there is a non-empty
set $Y \in \I$ such that $f(Y)<f(X^*)$, then this set $Y$ must
separate $t$ and $u$ for some pendant pair $(t,u)$ found in the
execution of the algorithm. Indeed, suppose that this was not the
case. Then, by induction, for every element $v$ of the ground set at a
particular iteration, the associated set $X_v$ of elements in the
original ground set that have been contracted into $v$ so far, is
always either completely inside $Y$ or completely outside $Y$. In
particular, in the last iteration $X_s$ must always be outside
$Y$. Therefore, at the end of the algorithm, $Y$ must be equal to the
set $X_t$ defined in line \ref{line-last} and so it is included in the
set of candidates, contradicting the definition of $X^*$. Consider
then the first pendant pair $(t,u)$ separated by $Y$. By the property
of pendant pairs, $f'(\{u\}) \leq f(Y)$ for the function $f'$ at that
iteration. But then, the set $X_u \in V$ of elements that were
contracted to $u$ is a candidate considered by the
algorithm. Therefore $f(X^*)\leq f(X_u)=f'(\{u\})\leq f(Y)$, which
contradicts our assumption.

Furthermore, since we choose $X^*$ as the set that is introduced first
into the family of candidates $\C$ (among the ones of minimum value),
then this set $X^*$ is also be a minimal optimal solution of
$(V,f,\I)$. Indeed, if there is a set $Y\in \I$ such that
$f(Y)=f(X^*)$, with $\emptyset\neq Y\subset X^*$, then this set must
separate two elements of $X^*$. This means that at some moment before
the introduction of $X^*$ as a candidate, the algorithm finds a
pendant pair $(t,u)$ separated by the set $Y$ with both $t,u \in
X^*$. At this iteration, the candidate $X_u$ introduced is such
that $f(X_u)=f(Y)=f(X^*)$, which is a contradiction since $X_u$ is
introduced earlier than $X^*$ to the set of candidates.

In order to achieve $O(|V|^3)$ function value oracle calls we don't
compute the functions $f'$ explicitly, but instead we keep track of
the partition of $V$ induced by the contraction of the elements.  By
using the fact that each iteration decreases the cardinality of $V'$
by one or two units and Lemma \ref{lem-pendantpair} we obtain the
desired bound on the number of function value oracle calls. This
completes the proof of Theorem \ref{thm-sym-alg} (and hence of Theorem
\ref{thm-sym} as well).

We can use the fact that the minimal solutions are disjoint to find
all minimal solutions. We first compute one particular minimal
solution $X^*$ of the system and contract it into a single element $s$
which we will consider a loop for the new family. Then we run the
algorithm again in such a way that, every time a minimal solution $X$
is found we contract $X+s$ into $s$ in order to avoid finding
solutions containing $X$ after that. The procedure is described in
Algorithm~\ref{algo2}.

\begin{algorithm}[ht]
\caption{FindMinimals $(V,f,\I)$}
\label{algo2}
\begin{algorithmic}[1]
\Require{A submodular system $(V,f,\I)$ where $f$ is symmetric and crossing submodular and $\I$ is not trivial.}
\Ensure{The family $\F$ of minimal optimal solutions for the hereditary minimization problem.}
\State{Compute, using \textrm{FindOptimal}, a minimal optimal solution $X^*$ for the system. Let $\lambda^*=f(X^*)$.}
\State{Let $(V'+s,f',\I')$ be the system obtained by contracting $X^*$ and all the loops of $\I$ into a single element, denoted $s$. (During the execution of the algorithm, we keep $s$ as an element \emph{outside} $V'$.)}
\State{$\I' \gets \I' \setminus \{A \in \I': s \in A\}$. \Comment{If $s$ is not a loop, we consider it as one.}}
\State{Let $\F=\{X^*\}$.}
\For {each $v \in V$ with $f'(\{v\})=\lambda^*$}
    \State {Add $\{v\}$ to $\F$}
    \State{Update $(V'+s,f',\I')$ by contracting $\{s,v\}$ into $s$.}
\EndFor
\While{$|V'|\geq 2$}\Comment{$f'(\{v\})>\lambda^*$ for all $v \in V'$, and $s$ is the only loop of $\I'$.}
\State{Find a pendant pair $(t,u)$ of $f'$ not containing $s$.}
\If{$\{t,u\} \in \I'$ and $f'(\{t,u\})=\lambda^*$}
    \State{Add $X_t\cup X_u$ to $\F$.}
    \State{Update $(V'+s,f',\I')$ by contracting $\{s,t,u\}$ into $s$.}
\ElsIf{$\{t,u\} \in \I'$ and $f'(\{t,u\})>\lambda^*$}
    \State{Update $(V'+s,f',\I')$ by contracting $\{t,u\}$ into $t$.}
\Else \Comment{$\{t,u\} \not\in \I'$.}
    \State{Update $(V'+s,f',\I')$ by contracting $\{s,t,u\}$ into $s$.}
\EndIf
\EndWhile
\State{Return the family $\F$.}
\end{algorithmic}
\end{algorithm}

\begin{thm}\label{thm-minimal-alg} Given a symmetric and crossing
submodular function $f$ on $V$, and an hereditary family $\I$ of
subsets of $V$, Algorithm \ref{algo2} outputs the collection of all
minimal optimal solutions for the associated hereditary minimization
problem in $O(|V|^3)$ function value oracle calls.  \end{thm}

By the previous discussion, we can see that every set added to $\F$
during the execution of this algorithm is a minimal optimal solution
of $(V,f,\I)$. We only need to show that no other minimal optimal
solution exists. Assume that this is not the case, i.e.~that there is
a nonempty set $Y \in \I$ that is a minimal optimal solution of
$(V,f,\I)$ with $Y \not\in \F$.

We first claim that at every moment and for every $v \in V'+s$, the associated set $X_v$ is always completely inside or completely outside $Y$. We prove this by induction. The claim is true at the beginning of the algorithm, and immediately after all the optimal singletons are added to $\F$ and contracted into~$s$. Suppose that the claim holds at the beginning of an iteration in the while-loop and let $(t,u)$ be the pendant pair found at that moment.  We note that $Y$ can't separate $t$ from $u$, since in that case we would have $f'(\{u\})=f(Y)=\lambda^*$. But, by construction, the algorithm ensures that at every iteration the singletons are not optimal, i.e., $f'(\{v\})>\lambda^*$ for every $v \in V'$. It follows that both $X_t$ and $X_u$ are either completely inside or completely outside $Y$. If all the elements participating in a contraction at this iteration are completely inside or completely outside $Y$ then the claim will still hold at the end of the iteration. The only case left to consider is that $X_t \cup X_u \subseteq Y$, $X_s \subseteq V\setminus Y$ and we contract $\{s,t,u\}$ into $s$. We only do this when $\{t,u\}\in \I'$ and $f'(\{t,u\})=\lambda^*$ or when $\{t,u\} \not\in \I'$. Since $Y \in \I$, we must be in the first case, and so, according to the algorithm, $X_{t,u}=X_t \cup X_u$ gets added to $\F$. By minimality of $Y$ we obtain $Y=X_{t,u}$ which contradicts the fact that $Y \not\in \F$. This proves the claim.

Since $Y$ is never added to $\F$, and $X_s \supseteq X^*$ is completely outside $Y$, the previous claim implies that after the while-loop, the set $Y$ must correspond to the unique element in $V'$, say $Y=X_t$, for $V'=\{t\}$. But by construction, we know that a singleton cannot be optimal, reaching a contradiction. This proves the correctness of the algorithm and, using the fact that both algorithms presented compute pendant pairs $O(|V|)$ times, it also completes the proof of Theorem \ref{thm-minimal-alg}.

\section{Extensions} \label{sec:ext}

We observe here that the proof of correctness of the first
algorithm relies only on the fact that we can find pendant pairs not
containing a particular element $s$ in each iteration. The second
algorithm also needs that minimal optimal solutions are disjoint. We
can use this to generalize the previous results to wider classes of
functions.

Given a set function $f$ on $V$, and a partition $\Pi$ ($=\{V_1, V_2,
\cdots, V_k\}$) of $V$, we define the \emph{fusion of $f$ relative to
$\Pi$} (also called the induced set function on $\Pi$), denoted by
$f_\Pi$, to be the function defined on subsets $X\subseteq \Pi$ by
$$f_\Pi(X)= f \biggl(\bigcup_{S \in X}S\biggr).$$ 
We say that a set function $f$ on $V$ is \emph{admissible} if for
every partition $\Pi$ of $V$ in at least three parts, and for every $S
\in \Pi$, the function $f_\Pi$ admits a pendant pair (defined as
before) avoiding $S$,
that is, a pendant pair $(T,U)$ with $S \not\in\{T,U\}$. We observe
that if $f$ is a symmetric crossing submodular function on $V$, so are
all the functions induced by partitions. Lemma \ref{lem-pendantpair}
says not only that symmetric crossing submodular functions are
admissible (a fact originally proven by Mader~\cite{mader_ber_1972}),
but that for every induced function we can find such a pendant pair
efficiently: for each $f_\Pi$ we use $O(|\Pi|^2)$ oracle calls to
$f$. 

The discussion at the beginning of this section implies the following result.

\begin{thm}\label{thm-extension} Given an hereditary
family $\I$ on $V$ and an admissible function $f$ on $V$ such that for
any partition $\Pi$ of $V$ in at least three parts and for every $S
\in \Pi$ we can find a pendant pair avoiding $S$ using $T(|\Pi|)$
calls to some oracle. Then, there is an algorithm that finds a minimal
optimal solution for the associated hereditary minimization problem
using $O(|V|\cdot T(|V|))$ oracle calls. If we can further ensure that
minimal solutions are disjoint, then we can find all minimal solutions
using $O(|V|\cdot T(|V|))$ oracle calls.  \end{thm}

Rizzi~\cite{rizzi_note_2000} exhibits a wider class of admissible
functions for which pendant pairs can be found efficiently. Consider a
real valued map $d$ defined on pairs of disjoint subsets of $V$, that
satisfies: \begin{enumerate} \item {\bf Symmetry:} $d(A,B)=d(B,A)$ for every
$A,B$ disjoint.  \item {\bf Monotonicity:} $d(A,B) \leq d(A,B\cup W)$ for
every $A,B,W$ pairwise disjoint.  \item {\bf Consistency:} $d(A,W) \geq
d(B,W)$ implies that $d(A, W \cup B) \geq d(B, W \cup A)$ for every
$A,B,W$ pairwise disjoint.  \end{enumerate}

For example, if $G$ is a graph with vertex set $V$ then the function
$d(A,B)$ defined as the weight of the edges having one endpoint in $A$
and the other in $B$ satisfies the previous properties. More
generally, if $f$ is a symmetric crossing submodular function then the
function $d(A,B) = \frac{1}{2}\left(f(A)+f(B) - f(A \cup B)\right)$ is
symmetric, monotone and consistent. The following more interesting
example shows that if $d$ is symmetric, monotone and consistent then
$f(S)=d(S,V\setminus S)$ is not necessarily a (crossing) submodular
function.  This example is given by Rizzi. Given a weighted graph $G$
on $V$, let $\lambda(u,v)$ be the shortest path distance between $u$
and $v$, and define the function $d(A,B)$ as the maximum value of
$\lambda(u,v)$ for $u \in A$ and $b\in B$. It is easy to check that
this map is symmetric, monotone and consistent. The coresponding
function $f$ given by $f(S)=d(S,V\setminus S)$ is, however, not
crossing submodular. Consider indeed the 4-cycle
$(V,E)=\{\{a,b,c,d\},\{ab,bc,cd,da\}\}$ with unit weights. We have
$3=f(\{a,c\})+f(\{a,d\}) < f(\{a,c,d\})+f(\{a\})=4$, and so it is not
(crossing) submodular.

In our terminology, Rizzi shows that for every such function $d$, the
set function $f(S)=d(S,V\setminus S)$ is admissible and we can find a
pendant pair avoiding any element by using a procedure similar to
Queyranne's. This procedure uses $O(|V|^2)$ oracle calls for $d$. Our
theorem then implies that we can find one minimal minimizer for an
admissible function $f$ constrained to a hereditary family using
$O(|V|^3)$ oracle calls.

Our second algorithm to find all minimal optimal solutions when
restricted to a hereditary family also
applies to set functions $f$ arising from a symmetric, monotone and
consistent map $d$, as we can argue that the minimal optimal solutions
are disjoint.

\begin{lem}
Let $d$ be a symmetric, monotone and consistent map on $V$ as defined
above, and let $f$ be defined by $f(S)=d(S,V\setminus S)$ for all
$S\subset V$. Let $\I$ be an hereditary family of subsets of $V$. Then
the minimal minimizers of $f$ constrained to $\I$ are disjoint. 
\end{lem}

\begin{proof}
Let $S$ and $T$ be two intersecting minimal minimizers of $f$ over
$\I$. Since $\I$ is hereditary, $S\setminus T$ and $T\setminus S$ are
also in $\I$, and by minimality, we have that $f(S\setminus T)>f(S)$
and $f(T\setminus S)>f(T)$. 

By the consistency assumption applied to $A=V\setminus T$,
$B=T\setminus S$ and $W=S\cap T$, we derive that $d(B,W)>d(A,W)$,
i.e. $d(T\setminus S, S\cap T)>d(V\setminus T, S\cap T)$. Furthermore,
by monotonicity, we get $d(V \setminus T, S\cap T)\geq d(S\setminus T,
S\cap T)$ implying that 
$$d(T\setminus S, S\cap T)> d(S\setminus T, S\cap T).$$
But, by exchanging the roles of $S$ and $T$, we get the reverse
inequality, contradicting the fact that two minimal minimizers can
intersect. 
\end{proof}

Another line of generalization is the one proposed by Nagamochi and
Ibaraki~\cite{nagamochi_noteminimizing_1998} and
Narayanan~\cite{narayanan_noteminimization_2003}. They consider
\emph{restrictions} of symmetric and crossing submodular
functions. Note that if $h$ is a symmetric and crossing submodular
function on $V$, and $T$ is a nonempty subset $T$ of $V$, then the
restriction $f$ of $h$ to the set $T$ defined as $f(X) = h(X)$ for all
$X \subseteq T$ is intersecting submodular (i.e.~the submodular
inequality holds for any pair of intersecting sets $A, B\subseteq V$,
that is sets  with $A\setminus B\neq
\emptyset$, $B\setminus A\neq \emptyset$, and $A\cap B\neq
\emptyset$), but it is not necessarily symmetric. However, it still
satisfies a weaker property known as intersecting posimodularity, that
is: \begin{equation}\label{eqn-posimodularity} f(A \setminus B) + f(B
\setminus A) \leq f(A) + f(B), \end{equation} for every pair of
intersecting subsets $A$ and $B$ of $V$. In fact, it is easy to see
that any intersecting submodular and intersecting posimodular function
can be obtained as a restriction of a symmetric crossing submodular
function (see \cite{narayanan_noteminimization_2003}). To be precise,
if $f$ is intersecting submodular and intersecting posimodular on $V$,
and $s$ is an element outside $V$, then the \emph{antirestriction}
function $g$ on $V+s$ defined as:
\begin{align}\label{antirestriction1}
  g(X) &= \begin{cases}
    f(X), &\text{if } s \not\in X\\
    f(V \setminus X), &\text{if } s \in X.
  \end{cases}
\end{align}
is a symmetric and crossing submodular function on $V + s$. Note also that for any hereditary family $\I$ on $V$, the set of optimal solutions of the system $(V,f,\I)$ is the same as the set of optimal solutions of $(V+s,g,\I)$, and so, we can find all the minimal minimizers of the original system by applying our algorithms to the second one.  This type of functions appears very often, for example, the sum of a symmetric submodular function with a modular function is clearly posimodular but it is not necessarily symmetric.

It is worth noting at this point that we can also use our methods to find all the inclusionwise maximum minimizers of \emph{contractions} (in the submodular sense) of symmetric and crossing submodular functions constrained to \emph{co-hereditary families} (closed under union). Given a symmetric crossing submodular function $h$ on $V$, and a nonempty set $T \subseteq V$, the contraction $f$ of $h$ to the nonempty subset $T$ of $V$ is defined as $f(X) = h(X \cup (V\setminus T)) - h(V\setminus T)$. Then, it is easy to see that the function $\bar{f}: T \to \R$ defined as $\bar{f}(X)=f(T\setminus X)$ is intersecting submodular and intersecting posimodular since $\bar{f}(X)=f(T\setminus X)=h(T\setminus X \cup (V\setminus T))-h(V\setminus T)=h(V\setminus X) -h(V\setminus T)= h(X) - h(V\setminus T)$. And so, in order to find all the maximum minimizers of $f$ under a co-hereditary family $\I$ of $T$ we can simply find the minimum minimizers of $\bar{f}$ under the hereditary family formed by the complements of the sets in $\I$.

\paragraph{Acknowledgment:} We would like to thank Shaddin Dughmi and Jan Vondrak for introducing us to the problem of minimizing submodular functions under cardinality constraints and for very useful discussions.

\bibliographystyle{abbrv}

\end{document}